\documentclass[journal,onecolumn,draftcls]{IEEEtran}

\usepackage{amssymb}
\usepackage{amsmath,amsthm}
\usepackage{graphicx}
\usepackage{graphics}
\newtheorem{theorem}{Theorem}
\newtheorem{corollary}{Corollary}
\usepackage{setspace}
\usepackage{amsmath}

\begin{document}
	
	\title{Reliability of k-out-of-n Data Storage System with Deterministic Parallel and Serial Repair}
	
	\author{Vaneet Aggarwal\thanks{V. Aggarwal is with the School of IE, Purdue University, email: vaneet@purdue.edu}}
	\maketitle

\begin{abstract}
	In this paper, we find the Laplace Stieltjes transform of the probability of data loss for the k-out-of-n distributed storage system with deterministic repair times. We consider two repair models, namely the serial and parallel repair. We show that for failure rate much lower than the repair rate, mean time of data loss for the two models is the same unlike the case for exponential repair models.
\end{abstract}


\section{Introduction}

Data storage demands have been increasing very rapidly, leading to a lot of interest in distributed storage systems \cite{computer,computer2}. Erasure-coded distributed storage systems provide reliable data storage at a a fraction of the cost as compared to the systems based on simple data replication.  Characterizing the reliability of the storage system is an important problem. Reliability has been explored for exponential repair times for serial repair \cite{Chen} and parallel repair \cite{Angus} models. Recently, there has been interest to characterize the reliability for deterministc repair times \cite{VinayRel,VinayRel2}, where geometric interpretations are provided. However, there is no complete characterization of reliability for parallel and serial repair models. The authors of \cite{MoustagaGrep} considers availability of the k-out-of-n systems with general repair times while this paper considers the reliability of k-out-of-n systems.   In this paper, we give the Laplace Stieltjes transform of the probability of data loss for both parallel and serial repair models.

Using the Laplace Stieltjes Transform formulas, mean time to data loss can be found. Even though expressions are difficult to evaluate exactly, we can find the dominant terms for the case when repair rate is much higher than the failure rate and find that the dominant term for the serial repair is the same as in the case of exponential repair in \cite{Chen}. However, we also find that the parallel repair also gives the same dominant term as the serial repair, and thus not achieving the $(n-k)!$ factor improvement in mean time to data loss as given in \cite{Angus} for the parallel exponential repair. The results are compared in simulations, where we see that the mean time to data loss with parallel repair is higher, but the asymptote as the repair rate increases is the same.

The rest of the paper is organized as follows. In Section 2, we give our results for the serial repair model. Section 3 presents our results for the parallel repair model. Section 4 gives the numerical results and Section 5 concludes this paper.

\section{Serial Repair Model}

The serial repair model can be described with the following assumptions. \begin{enumerate}
\item	$n>0$ components (disks) are identical and independent.
\item	The failure rate for each disk is constant, and equals $\lambda$ (failure time is exponentially distributed for each disk).
\item	The repair times are independent and fixed, and equals $t_{rep}$.
\item	The repair for the previous broken components must be re-initialized when a new component fails, and only the newest broken component is repaired.
\item	The system stops working when  $n-k+1$ components have failed for $0<k<n$, $k,n\in {\mathbb{Z}}$.
\end{enumerate}

Suppose that the probability of data loss in time $t$ is $P_c(t)$. Then the Laplace Stieltjes Transform of reliability is given as $P_c(s) = {\mathbb{E}}[e^{-sT}]$, where $T$ is the random variable representing the time to data loss, and $P_c(t) = \Pr(T\le t)$. The next result gives the Laplace Stieltjes Transform of reliability for serial repair.

\begin{theorem}
The Laplace Stieltjes Transform of reliability for serial repair is given by
\begin{equation}
\resizebox{\linewidth}{!}{$\displaystyle
P_c(s) = \frac{Q_{0,1}^*(s)Q_{1,2}^*(s)Q_{2,3}^*(s)\cdots Q_{n-k,n-k+1}^*(s)}{\det\left(\begin{bmatrix}
1 & -Q_{0,1}^*(s) & 0 &  \cdots & 0 & 0& 0\\
-Q_{1,0}^*(s) & 1 & -Q_{1,2}^*(s) & \cdots & 0 & 0& 0\\
\cdots\\
\cdots\\
0 & 0 & 0 & \cdots & -Q_{n-k-1,n-k-2}^*(s) & 1 & -Q_{n-k-1,n-k}^*(s)\\
0& 0 & 0 & \cdots & 0 & -Q_{n-k,n-k-1}^*(s) & 1
\end{bmatrix}
\right)}$,}
\end{equation}
where
\begin{equation}
Q_{ij}^*(s) = \begin{cases}\frac{n\lambda}{s+n\lambda} & \text{ if } i=0, j=1 \\
\frac{(n-i)\lambda}{(n-i)\lambda+s}(1-e^{-((n-i)\lambda+s) t_{rep}}) & \text{ if }  j=i+1, 1\le i\le n-k\\
e^{-((n-i)\lambda+s) t_{rep}} & \text{ if }  j=i-1, 1\le i\le n-k\\
0 & \text{ otherwise} \end{cases}
\end{equation}
\end{theorem}


\begin{proof}
Let $Q_{ij}(t)$ be the one-step transition probability from state $i$ to state $j$ in time $t$, in which one-step implies that there is only one repair or failure during an interval of time $t$. Then, $Q_{ij}(t)$ is given as

\begin{equation}
Q_{ij}(t) = \begin{cases}\int_{x=0}^tn\lambda e^{-n\lambda x} dx & \text{ if } i=0, j=1 \\
\int_{x=0}^{\min(t,t_{rep})}(n-i)\lambda e^{-(n-i)\lambda x} dx & \text{ if }  j=i+1, 1\le i\le n-k\\
\int_{x=0}^{t} e^{-(n-i)\lambda t_{rep}} \delta(x=t_{rep}) dx & \text{ if }  j=i-1, 1\le i\le n-k\\
0 & \text{ otherwise} \end{cases}
\end{equation}

The Laplace-Stieltjes transform is given by
\begin{equation}
Q_{ij}^*(s) = \begin{cases}\frac{n\lambda}{s+n\lambda} & \text{ if } i=0, j=1 \\
\frac{(n-i)\lambda}{(n-i)\lambda+s}(1-e^{-((n-i)\lambda+s) t_{rep}}) & \text{ if }  j=i+1, 1\le i\le n-k\\
e^{-((n-i)\lambda+s) t_{rep}} & \text{ if }  j=i-1, 1\le i\le n-k\\
0 & \text{ otherwise} \end{cases}
\end{equation}
Let $H_{ij}(t)$ be the first recurrence time from state $i$ to state $j$. Then, $H_{ij}(t)$ is given as

\begin{equation}
H_{ij}(t) = Q_{ij}(t) + \sum_{m\ne j} \int_{x=0}^{t} H_{m,j}(t-x)dQ_{i,m}(x)
\end{equation}
The Laplace-Stieltjes transform is then given by
\begin{equation}
H_{ij}^*(s) = Q_{ij}^*(s) + \sum_{m\ne j} H_{m,j}^*(s)Q_{i,m}^*(s)
\end{equation}

The above can be used to write
\begin{equation}
\resizebox{\linewidth}{!}{$\displaystyle
\begin{bmatrix}
H_{0,n-k+1}^*(s)\\
H_{1,n-k+1}^*(s)\\
H_{2,n-k+1}^*(s)\\
\cdots\\
H_{n-k-2,n-k+1}^*(s)\\
H_{n-k-1,n-k+1}^*(s)\\
H_{n-k,n-k+1}^*(s)
\end{bmatrix}
=
\begin{bmatrix}
0 & Q_{0,1}^*(s) & 0 &  \cdots & 0 & 0& 0\\
Q_{1,0}^*(s) & 0 & Q_{1,2}^*(s) & \cdots & 0 & 0& 0\\
\cdots\\
\cdots\\
0 & 0 & 0 & \cdots & Q_{n-k-1,n-k-2}^*(s) & 0 & Q_{n-k-1,n-k}^*(s)\\
0& 0 & 0 & \cdots & 0 & Q_{n-k,n-k-1}^*(s) & 0
\end{bmatrix}
\begin{bmatrix}
H_{0,n-k+1}^*(s)\\
H_{1,n-k+1}^*(s)\\
H_{2,n-k+1}^*(s)\\
\cdots\\
H_{n-k-2,n-k+1}^*(s)\\
H_{n-k-1,n-k+1}^*(s)\\
H_{n-k,n-k+1}^*(s)
\end{bmatrix}
+ \begin{bmatrix}
0\\0\\0\\ \cdots\\
0\\0\\Q_{n-k,n-k+1}^*(s)\end{bmatrix}$}
\label{mateqn}\end{equation}

Let $P_i(t)$ be the probability that the system is at state $i$ at time $t$ starting from state $0$ at time $t=0$, and let $P_i^*(s)$ be the Laplace-Steiltjes transform. Then, the probability of data loss is given by $P_{n-k+1}(t) = H_{0,n-k+1}(t)$. Thus, it is enough to find $H_{0,n-k+1}(t)$.
From Equation \eqref{mateqn}, we have

\begin{equation}
\resizebox{\linewidth}{!}{$\displaystyle
H_{0,n-k+1}^*(s) = \frac{Q_{0,1}^*(s)Q_{1,2}^*(s)Q_{2,3}^*(s)\cdots Q_{n-k,n-k+1}^*(s)}{\det\left(\begin{bmatrix}
1 & -Q_{0,1}^*(s) & 0 &  \cdots & 0 & 0& 0\\
-Q_{1,0}^*(s) & 1 & -Q_{1,2}^*(s) & \cdots & 0 & 0& 0\\
\cdots\\
\cdots\\
0 & 0 & 0 & \cdots & -Q_{n-k-1,n-k-2}^*(s) & 1 & -Q_{n-k-1,n-k}^*(s)\\
0& 0 & 0 & \cdots & 0 & -Q_{n-k,n-k-1}^*(s) & 1
\end{bmatrix}
\right)}$}
\end{equation}
\end{proof}

Using $P_c(s)$, we can find the mean time to data loss as $-P_c'(0)$. Given the exact expression of $P_c(s)$, we can evaluate $-P_c'(0)$. However, the expression is cumbersome and thus omitted. We consider however the scenario where $\lambda t_{rep}<<1$ and give the mean time to data loss in that scenario.

\begin{corollary}
	The mean time to data loss for serial deterministic repair for $\lambda t_{rep}<<1$ is approximately given as
	\begin{equation}
	{\text {MTDL}}_c \approx \frac{(k-1)!}{n!\lambda}\frac{1}{(\lambda t_{rep})^{n-k}}
	\end{equation}
\end{corollary}
\begin{proof}
	The proof follows by using expression for $-P_c'(0)$ and ignoring all terms of lower order. The detailed steps are omitted.
\end{proof}

We note that for exponential serial repair model with rate of repair $\mu$, mean time to data loss was characterized in \cite{VaneetOpp}, and is given as

\begin{equation}
{\text {MTDL}}_{Chen} = \sum_{l=0}^{n-k}\frac{1}{(n-l)!}\sum_{i=0}^{n-k-l} \mu^i \lambda^{-(i+1)}(n-l-i-1)!,
\end{equation}
with an approximation given as \cite{Chen}
\begin{equation}
{\text {MTDL}}_{Chen} \approx \frac{(k-1)!}{n!}\frac{\mu^{n-k}}{\lambda^{n-k+1}}.
\end{equation}
We note that the MTDL approximation for serial repair with deterministic repair is the same as the approximation with exponential repair for $t_{rep}= 1/\mu$.

\section{Parallel Repair Model}
The parallel repair model can be described with the following assumptions. \begin{enumerate}
	\item	$n>0$ components (disks) are identical and independent.
	\item	The failure rate for each disk is constant, and equals $\lambda$ (failure time is exponentially distributed for each disk).
	\item	The repair times are independent and fixed, and equals $t_{rep}$.
	\item	The repair for the previous broken components must be re-initialized when a new component fails, and all the  broken components are simultaneously repaired.
	\item	The system stops working when  $n-k+1$ components have failed for $0<k<n$, $k,n\in {\mathbb{Z}}$.
\end{enumerate}

Suppose that the probability of data loss in time $t$ is $P_a(t)$. Then the Laplace Stieltjes Transform of reliability is given as $P_a(s) = {\mathbb{E}}[e^{-sT}]$, where $T$ is the random variable representing the time to data loss, and $P_a(t) = \Pr(T\le t)$. The next result gives the Laplace Stieltjes Transform of reliability for parallel repair.

\begin{theorem}
	The Laplace Stieltjes Transform of reliability for parallel repair is given by
	\begin{equation}
	\resizebox{\linewidth}{!}{$\displaystyle
		P_a(s) = \frac{Q_{0,1}^*(s)Q_{1,2}^*(s)Q_{2,3}^*(s)\cdots Q_{n-k,n-k+1}^*(s)}{\det\left(\begin{bmatrix}
			1 & -Q_{0,1}^*(s) & 0 &  \cdots & 0 & 0& 0\\
			-Q_{1,0}^*(s) & 1 & -Q_{1,2}^*(s) & \cdots & 0 & 0& 0\\
			\cdots\\
			\cdots\\
			-Q_{n-k-1,0}^*(s) & 0 & 0 & \cdots & 0 & 1 & -Q_{n-k-1,n-k}^*(s)\\
			-Q_{n-k,0}^*(s)& 0 & 0 & \cdots & 0 & 0 & 1
			\end{bmatrix}
			\right)}$,}
	\end{equation}
	where
	\begin{equation}
	Q_{ij}^*(s) = \begin{cases}\frac{n\lambda}{s+n\lambda} & \text{ if } i=0, j=1 \\
	\frac{(n-i)\lambda}{(n-i)\lambda+s}(1-e^{-((n-i)\lambda+s) t_{rep}}) & \text{ if }  j=i+1, 1\le i\le n-k\\
	e^{-((n-i)\lambda+s) t_{rep}} & \text{ if }  j=0, 1\le i\le n-k\\
	0 & \text{ otherwise} \end{cases}
	\end{equation}
	
	The above can be simplified to
	
	\begin{equation}
	P_a(s) = \frac{Q_{0,1}^*(s)Q_{1,2}^*(s)Q_{2,3}^*(s)\cdots Q_{n-k,n-k+1}^*(s)}{1-\sum_{j=1}^{n-k}\left((\Pi_{i=0}^{j-1}Q_{i,i+1}^*(s))Q_{j,0}^*(s)\right)}.
	\end{equation}

\end{theorem}

\begin{proof}
	The proof steps are very similar to that for serial repair model, where the difference is that all the failed components get repaired simultaneously thus going from State i to State 0 on repair rather than to State i-1 for the serial repair model.
\end{proof}

Using $P_a(s)$, we can find the mean time to data loss as $-P_a'(0)$. Given the exact expression of $P_a(s)$, we can evaluate $-P_a'(0)$. However, the expression is cumbersome and thus omitted. We consider however the scenario where $\lambda t_{rep}<<1$ and give the mean time to data loss in that scenario.

\begin{corollary}
	The mean time to data loss for parallel deterministic repair for $\lambda t_{rep}<<1$ is approximately given as
	\begin{equation}
	{\text {MTDL}}_a \approx \frac{(k-1)!}{n!\lambda}\frac{1}{(\lambda t_{rep})^{n-k}}
	\end{equation}
\end{corollary}
\begin{proof}
	The proof follows by using expression for $-P_a'(0)$ and ignoring all terms of lower order. The detailed steps are omitted.
\end{proof}

We note that the approximation is the same, thus giving the same asymptotic expression as in the case of serial repsir model. For exponential parallel repair model with rate of repair $\mu$, mean time to data loss was characterized in \cite{VaneetOpp}, and is given as

\begin{equation}
{\text {MTDL}}_{Angus} = \sum_{l=0}^{n-k}\frac{1}{(n-l)!}\sum_{i=0}^{n-k-l} \mu^i \lambda^{-(i+1)}(n-l-i-1)!i!,
\end{equation}
with an approximation given as \cite{Chen}
\begin{equation}
{\text {MTDL}}_{Angus} \approx \frac{(k-1)!}{n!}\frac{\mu^{n-k}}{\lambda^{n-k+1}} (n-k)!.
\end{equation}
Thus, even though exponential parallel repair has $(n-k)!$ times higer MTDL than serial repair, the MTDL approximations for the serial and parallel deterministic repair models is the same.

\section{Numerical Results}
\begin{figure}[hbtp]
	\centering
	\includegraphics[trim=.5in 2.5in 1in 3in, clip, width=0.8\textwidth]{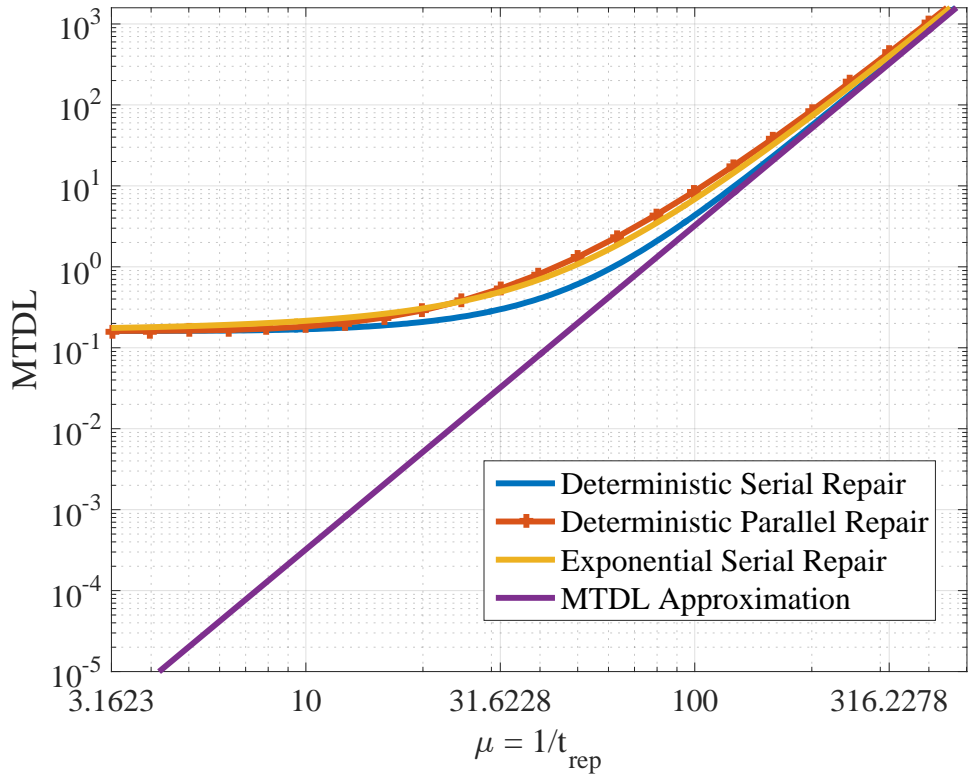}
	\caption{Mean time to data loss for serial and parallel deterministic repair models, with comparison to the exponential serial repair and the approximation result.}
	\label{fig:awesome_image}
\end{figure}
We assume $n=10$ and $k=6$. We assume that the failue rate is $4$ per year. We plot the increase in mean time to data loss with decreasing time to repair in Figure \ref{fig:awesome_image}. We note that for serial repair model, the mean time to data loss is lower for the deterministic repair model as compared to the serial exponential repair model. We further see that parallel deterministic repair has higher mean time to data loss as compared to serial deterministic repair. Finally, we note that the approximation result approximates all these three mean time to data loss expressions as the repair time becomes smaller.

\section{Conclusions}

	This paper gives Laplace Stieltjes transform of the  data loss for the k-out-of-n distributed storage system with deterministic repair times. Two repair models, namely the serial and parallel repair are considered. We show that for failure rate much lower than the repair rate, mean time of data loss for the two models is the same unlike the case for exponential repair models.

\end{document}